\documentclass{ifacconf}

\usepackage{graphicx}      
\usepackage{natbib}        

\usepackage{todonotes}
\usepackage{amsmath} 
\usepackage{amssymb}
\usepackage{enumitem}
\usepackage{siunitx}

\newcommand{\R}[1]{\ensuremath{\mathbb{R}^{#1}}}
\newcommand{\I}{\ensuremath{\mathbb{I}}}
\DeclareMathAlphabet{\mymathbb}{U}{BOONDOX-ds}{m}{n}
\newcommand{\Zero}{\ensuremath{\mymathbb{0}}}
\newcommand{\upright}[2]{\ensuremath{#1_{\mathrm{#2}}}}
\newcommand{\T}{\ensuremath{^{\mathrm{T}}}}

\newtheorem{assumption}{Assumption} 
\newtheorem{theorem}{Theorem} 
\newtheorem{proof}{Proof} %
\usepackage{pgfplots}
\usepackage{tikz}
\pgfplotsset{compat=newest}
\usetikzlibrary{plotmarks}
\usetikzlibrary{arrows.meta}
\usepgfplotslibrary{patchplots}
\usepackage{grffile}

\begin{document}
\begin{frontmatter}

\title{Model Predictive Control with Gaussian-Process-Supported Dynamical Constraints for Autonomous Vehicles}
\thanks[footnoteinfo]{The authors acknowledge support by the research training group DFG-GRK 2297 of the German Research foundation 'Deutsche Forschungsgemeinschaft'.}

\author[First]{Johanna Bethge} 
\author[Second]{Maik Pfefferkorn} 
\author[First]{Alexander Rose}
\author[Third]{Jan Peters}
\author[First]{Rolf Findeisen}

\address[First]{Control and Cyber-Physical Systems Laboratory, Technical University of Darmstadt, Darmstadt, Germany \newline \mbox{(e-mail:~\{johanna.bethge,~alexander.rose,~rolf.findeisen\}@iat.tu-darmstadt.de)}}
\address[Second]{Systems Theory and Automatic Control Laboratory, Otto-von-Guericke University Magdeburg, Magdeburg, Germany (e-mail: maik.pfefferkorn@ovgu.de)}
\address[Third]{Intelligent Autonomous Systems Laboratory, Technical University of Darmstadt, Darmstadt, Germany (e-mail:jan.peters@tu-darmstadt.de)}

\begin{abstract}                

We propose a model predictive control approach for autonomous vehicles  that exploits learned Gaussian processes for predicting human driving behavior.  
The proposed approach employs the uncertainty about the GP’s prediction to achieve safety.
A multi-mode predictive control approach considers the possible intentions of the human drivers. 
While the intentions are represented by different Gaussian processes, their probabilities foreseen in the observed behaviors are determined by a suitable online classification. 
Intentions below a certain probability threshold are neglected to improve performance. 
The proposed multi-mode model predictive control approach with Gaussian process regression support enables repeated feasibility and probabilistic constraint satisfaction with high probability. 
The approach is underlined in simulation, considering real-world measurements for training the Gaussian processes.

\end{abstract}


\end{frontmatter}

\section{Introduction}

Safe autonomous driving at intersections in mixed traffic (involving both human-driven and autonomous vehicles) remains challenging -- especially as predicting human driving behavior is nontrivial due to the wide range of individual behaviors.
One way to capture the complexity of human driving behavior is to learn models of human-driven vehicles from real data.
We employ Gaussian process (GP) regression \citep{Rasmussen2006} to model human-driven vehicles crossing traffic-light-free intersections.
More specifically, we use for each intention of the human-driven vehicles -- turning right, turning left and going straight on -- an independent GP-based model.
This enables us to learn a ``standard'' behavior of human drivers
in form of mean predictions of the GP and capture variations via the uncertainty measure of the GP.
The learned models enable us to predict the future position of a human-driven vehicle.

To enable the autonomous vehicle to safely cross the intersection in mixed traffic, we employ model predictive control (MPC) for decision making \citep{Rawlings2019, Findeisen2002}.
The proposed MPC scheme uses multiple GP-based models for predicting the possible actions of the target vehicle.
Specifically, the safety distance between human-driven and autonomous vehicle is adapted based on the model uncertainty.
This uncertainty enables to capture deviations between the actual behavior of the human-driven vehicle and standard human driving behavior.
Such adaptation ensures that a minimum safety distance is maintained at all times, despite the uncertainty in the predicted behavior of the target vehicle.
However, the challenge that the intention of the human-driven vehicle is a-priori unknown remains -- thereby introducing additional uncertainty.
To overcome this challenge, we predict the future positions of the human-driven vehicle for all likely maneuvers and ensure an adapted safety distance for each maneuver using a multi-mode MPC approach.
Supporting the multi-mode MPC scheme with an online classification, we estimate the probability of a specific maneuver being taken and remove unlikely modes - \textit{intentions} - from the MPC problem to overcome conservatism and to improve performance.

In a similar way, \cite{bethge2020modelling} used neural networks instead of Gaussian processes, which does not enable to capture uncertainty due to different driving behaviors.
\cite{batkovic2020_pedestrians} proposed an approach to avoid pedestrian obstacles exploiting multiple modes and uncertainty bounds.
Contrary to this work, the pedestrians cross the street orthogonally to the driving direction at a known location and that the safety distance is not adapted online.
While there are many articles on obstacle avoidance \citep{soloperto2019collision}, communication design for autonomous vehicles \citep{di2019design}, and optimization of traffic using traffic light control \citep{deSchutter1998optimal}, these topics are not the focus of this paper.


In summary, we propose a new GP-supported control scheme for safe intersection crossing of autonomous vehicles in mixed traffic.
We show how to derive GP-based models of human-driven vehicles at (traffic-light-free) intersections and how to exploit such models in a multi-mode MPC scheme.
The approach is supported by an online classification approach to estimate the intention of the human-driven vehicle and further exploits these estimate for improved performance. 
We guarantee safety, i.e., (robust) constraint satisfaction and recursive feasibility, of the proposed scheme with high probability.
In contrast to other works, we consider all (likely) intentions with corresponding, independent predictions of the behavior of the target vehicle simultaneously in our approach.
This has the potential to reduce conservatism when compared to other approaches, e.g., waiting for all target vehicles to have left the intersection.

The remainder of this paper is structured as follows. We introduce multi-mode MPC in Sec. \ref{sec:MultiModeMPC}.
In Sec. \ref{sec:GaussianProcessRegression}, we outline GP regression, including the derivation of the prediction models for the human-driven vehicles.
We illustrate our approach using a simulation example in Sec. \ref{sec:SimulationExample} and draw conclusions in Sec. \ref{sec:Conclusions}.



\section{Learning-supported Multi-mode MPC} \label{sec:MultiModeMPC}

We start by outlining the autonomous system (ego) dynamics, followed by a description of the models of human-driven (target) vehicles.
We finish this section by presenting a learning-supported multi-mode MPC scheme and derive probabilistic safety guarantees.

\subsection{Autonomous Ego Vehicle}
We consider the nonlinear time-invariant system that describes the ego vehicle dynamics:
\begin{subequations}
    \label{eq:SysModel_ego}
\begin{align}
    & x(k+1) = f(x(k),\,u(k)), \label{eq:SysModel_ego:stateMapping}\\
    & y(k) = C x(k). \label{eq:SysModel_ego:outputMapping}
\end{align}
\end{subequations}
Here, $x \in \mathcal{X} \subseteq \mathbb{R}^{n_x}$ is the system state, $y \in \mathcal{Y} \subseteq \mathbb{R}^{n_y}$ the output and $u \in \mathcal{U} \subseteq \mathbb{R}^{n_u}$ the control input of our autonomous ego vehicle. The state transition map is $f:\, \mathbb{R}^{n_x} \times \mathbb{R}^{n_u} \rightarrow \mathbb{R}^{n_x}$ and $h:\,\mathbb{R}^{n_x} \times \mathbb{R}^{n_u} \rightarrow \mathbb{R}^{n_y}$ denotes the output mapping.
We aim to design a predictive controller with state and output reference $x_\text{ref}$ and $y_\text{ref}$, respectively.
\begin{assumption}
    The reference is trackable, i.e., $x_\text{ref}$ and $y_\text{ref}$ can be exactly followed simultaneously at all time steps.
\end{assumption}

To overcome the challenge of introducing an online trajectory planner, we propose to use a (speed-assigned) path as reference \citep{Faulwasser2009, Matschek2019}.
The path $\mathcal{P}:\, s \mapsto (x_\text{ref},y_\text{ref})$ is a geometric curve, parameterized by the path parameter $s \in \R{}$, and defined independently of the time.
The path is equipped with a timing law of the form $s(k+1) = \alpha(s,\,u_s), ~ s \in [0,\, \infty), ~ u_s \in [\underline{u}_s,\,\bar{u}_s]$, where $u_s$ is a virtual control input bounded by $\underline{u}_s$ and $\bar{u}_s$.
By a suitable choice of the timing law, forward motion along the path is ensured.
The virtual input $u_s$ is used as an additional degree-of-freedom in the control scheme to actively decide about the reference timing.
The speed-assignment equips the path with a desired reference velocity, which is not strict and can be adapted by the controller for the benefit of low position errors.


\subsection{Human-Driven Target Vehicle}

The human-driven vehicle follows an a-priori unknown intention $i \in \mathcal{I}$ with
\begin{equation}
    \mathcal{I} = \{ \text{turning right},\, \text{turning left},\, \text{straight on} \}. \label{eq:Intention_general}
\end{equation}
%
To model the target vehicle, we employ for each intention $i$ an independent model of the form
\begin{align}
    \xi_i(k+1) = g_i(\xi_i(k)), \, i \in \mathcal{I} ~,  \label{eq:HumanModel_general}
\end{align}
with the measurable output $\xi_i \in \mathbb{R}^{n_\xi}$ of the target vehicle for a specific intention $i \in \mathcal{I}$.
Due to the wide range of individual driving behaviors of humans, it is nearly impossible to obtain an accurate model for an individual driver.
Rather, \eqref{eq:HumanModel_general} is designed to model the average behavior of human drivers for each intention and to provide an uncertainty measure to capture deviations of individual behaviors from the standard behavior, see Sec.~\ref{sec:humandrivermodel}.

As the intention of the target vehicle is initially unknown, we employ an online classification approach to determine in each time step $k$ a set of \textit{likely intentions} $\tilde{\mathcal{I}}_k \subseteq \mathcal{I}$.
Similar to \cite{bethge2020modelling}, we compare the observed target vehicle behavior to prototypical paths (i.e., the average paths for each intention) to quantify the probability of a specific intention.\footnote{Any other online classification scheme can be used as long as it provides a realistic estimate for the probability of the modes being active.}
At the initial time step $k=0$, we set $\tilde{\mathcal{I}}_0 = \mathcal{I}$.
In any subsequent time step, the classification algorithm exploits a sequence $\{ \xi \}(k)$ of previous observations of the target vehicle to provide the probability of intention $i$ being active, which is $\mathrm{P}(i ~|~ \{ \xi \}(k))$.
If $\mathrm{P}(i ~|~ \{ \xi \}(k)) < \upright{\epsilon}{P}$ for a design threshold $\upright{\epsilon}{P} \in (0,1)$, we update $\tilde{\mathcal{I}}_k$ by $\tilde{\mathcal{I}}_{k+1} = \tilde{\mathcal{I}}_k \setminus \{ i \}$.
\begin{assumption}
\label{ass:ClassificationProbability}
    The probabilities $\mathrm{P}(i ~|~ \{ \xi \}(k))$ provided by the online classification are $\forall i\in {\mathcal{I}}\setminus \tilde{\mathcal{I}}$ upper bounds of the true probabilities and lower bounds $\forall i \in \tilde{\mathcal{I}}$.
\end{assumption}

In addition to Assumption~\ref{ass:ClassificationProbability}, and for simplicity of notation, we assume in the following that $y \in \R{n_y}$ and $\xi \in \R{n_y}$ represent two-dimensional Cartesian positions $(p_x, p_y)$.
%
\subsection{Multi-Mode Model Predictive Controller} \label{sec:MPC:subsec:Controller}


We develop a multi-mode MPC scheme to safely navigate the ego vehicle over the intersection.
Such safe crossing requires keeping at least a minimal safety distance $\upright{d}{safe}$ to the target vehicle which depends on the velocity of the ego vehicle.
MPC is an optimal control scheme that is based on the repeated solution of a constrained finite-horizon optimal control problem (OCP) \citep{Rawlings2019, Findeisen2002}.
It can be described by
\begin{subequations}
    \label{eq:OCP_MMConstrLearn}
    \begin{align}
        \min_{u,u_s}  & \left \{ \sum_{j=0}^{N-1} l(\hat{x}(j), \hat{y}(j), u(j),\hat{s}(j)) \! + \! E(\hat{x}(N), \hat{y}(N), \hat{s}(N)) \right \} \nonumber \\
        \text{s.t.} ~ & \forall j \in \{0, \ldots, N-1 \}: \nonumber \\
        & ~~ \hat{x}(j+1) = f(\hat{x}(j),\,u(j)),~ \hat{x}(0) = x(k) \\
        & ~~ \hat{y}(j) =  C\cdot \hat{x}(j)\\
        & ~~ \forall i \in \tilde{\mathcal{I}}_k: \hat{\xi}_i(j+1) = g_i(\hat{\xi}_i(j)),~ \hat{\xi}_i(0) = \xi_i(k) \label{eq:OCP_MMConstrLearn:GPpred}\\
        & ~~ \forall i \in \tilde{\mathcal{I}}_k: \hat{d}_i(j) \geq d_\text{ref}(j) \label{eq:OCP_MMConstrLearn:distance} \\
        & ~~ \hat{x}(j) \in \mathcal{X}, \hat{y}(j) \in \mathcal{Y},\, u(j) \in \mathcal{U}, \hat{x}(N) \in \mathcal{X}_f\\
        & ~~ \hat{s}(j+1) = \alpha(\hat{s}(j), u_s(j)),~ \hat{s}(0) = s(k) \\
        & ~~ u_s(j) \in [\underline{u}_s,\, \bar{u}_s],\, \hat{s}(j) \in [0,\,\infty).
    \end{align}
\end{subequations}
Therein, $N$ denotes the horizon length, $\hat{\cdot}$ denotes a prediction, $l(\cdot)$ and $E(\cdot)$ are the stage and terminal cost, respectively, $\mathcal{X}_f$ is the terminal region, $\hat{d}_i(j) = \lVert \hat{y}(j) - \hat{\xi}_i(j) \rVert_2$ is the predicted distance between ego and target vehicle, and $\upright{d}{ref}(j)$ is the required minimum distance between both vehicles.
The latter one is composed of the minimal safety distance $\upright{d}{safe}(j)$, depending on the ego vehicle's speed, and an additional contribution $d_{\sigma,i}(j)$ accounting for the uncertainty in the target vehicle model.
Choosing $d_{\sigma,i}(j)$ according to a high-probability confidence interval of the target vehicle model, the actual vehicle distance is with high probability larger than the minimum safety distance $\upright{d}{safe}$.
From the optimal input sequence $u^*$, which solves \eqref{eq:OCP_MMConstrLearn}, the first element is applied to the system throughout the next sampling period.
Note that the controller is supported by the online classification scheme and accounts for all remaining likely intentions of the target vehicle.



\begin{theorem}
  \label{thm:ProabilityGuarantees}
  Let the above-mentioned assumptions hold.
  For the predictions of the target vehicle behavior, assume the confidence interval $\mathrm{P}(\lvert \xi_i(k) - \hat{\xi}_i(k) \rvert \leq \nu \sigma_i(k)) \geq 1 - \omega$, where $\xi_i(k)$ is the actual target vehicle position, $\hat{\xi}_i(k)$ is the prediction, $\omega \in (0,1)$ is a design parameter, $\nu$ is chosen in accordance with the confidence level $1-\omega$ and $\sigma_i(k)$ is the standard deviation describing the model uncertainty for intention $i$ at time step $k$.
  Then, constraint satisfaction and recursive feasibility of OCP \eqref{eq:OCP_MMConstrLearn} at time $k$ are given with probability $\mathrm{P}_\text{cs} (k) \geq 1 - \omega - \sum_{j\in \mathcal{I}(k) \setminus \tilde{\mathcal{I}}} \mathrm{P}(j \mid \{ \xi \}(k))$.
\end{theorem}
\begin{proof}
    The OCP \eqref{eq:OCP_MMConstrLearn} is a standard OCP, satisfying standard assumptions on MPC, except for \eqref{eq:OCP_MMConstrLearn:GPpred} and \eqref{eq:OCP_MMConstrLearn:distance}, 
    which introduce stochasticity due to the uncertain predictions of the target vehicle positions and the use of confidence-region-based constraint tightening.
    However, regarding all other constraints of the OCP \eqref{eq:OCP_MMConstrLearn}, standard guarantees as repeated feasibility and constraint satisfaction hold \citep{Rawlings2019}. 
    Given the $(1-\omega)$-confidence region for the predicted target vehicle positions, $d(k) \geq \upright{d}{safe}(k)$ is ensured with probability $1-\omega$ if $\hat{d}(k) \geq \upright{d}{safe}(k) + d_{\sigma}(k)$, where $d_{\sigma}(k)$ is chosen accordingly to the confidence region, c.f., \eqref{eq:OCP_MMConstrLearn:distance}.
    Violations of the minimal safety distance constraint between the ego and the vehicle might happen for two reasons: (i) the target vehicle behavior is not captured by the chosen high-probability confidence region obtained from the target vehicle model, or (ii) the true intention of the target vehicle is wrongly removed from the set of likely intentions and thus not considered by the controller anymore.
    For (i), each intention considered in the OCP (i.e., $\forall i \in \tilde{\mathcal{I}}_k$), the probability of constraint violation is by design given by $\omega$.
    As all $i \in \tilde{\mathcal{I}}_k$ are considered in \eqref{eq:OCP_MMConstrLearn:GPpred}, \eqref{eq:OCP_MMConstrLearn:distance}, i.e., they are considered equally likely by the controller, the total probability of constraint violation is not greater than $\omega$ if the true intention $i_*$ is for each time step $k$ contained in $\tilde{\mathcal{I}}_k$.
    If $i_* \notin \tilde{\mathcal{I}}_k$ for some time step $k$ (case (ii)), which happens with probability $\mathrm{P}(i_* ~|~ \{ \xi \}(k)) < \upright{\epsilon}{P}$ (Ass.~\ref{ass:ClassificationProbability}), the constraints are anyways not automatically violated and $\mathrm{P}(i_* ~|~ \{ \xi \}(k))$ is an upper bound for the probability of constraint violation in this case.
    Considering both cases, the probability for violating the constraints is given by $\mathrm{P}_{\text{cv}}(k) \leq \omega + \sum_{j\in \mathcal{I}(k) \setminus \tilde{\mathcal{I}}} \mathrm{P}(j ~|~ \{ \xi \}(k))$.
    Hence,  $\mathrm{P}_\text{cs} (k) \geq 1 - \mathrm{P}_{\text{cv}}(k)$ holds.
\end{proof}



\section{Gaussian Process Regression}
\label{sec:GaussianProcessRegression}

We first outline the basics of Gaussian process (GP) regression, followed by an overview of hyperparameter learning.
Thereafter, predictions at uncertain inputs are considered.
We finish by showing how we use GP models to predict human driving behavior in a model predictive controller for safe intersection crossing of an autonomous vehicle.

\subsection{Basics of Gaussian Process Regression}

A Gaussian process $\rho(\zeta) \sim \mathcal{GP}(m(\zeta), k(\zeta, \zeta')$ extends the concept of Gaussian probability distributions to random functions in order to model an uncertain function $\varphi: \R{d} \rightarrow \R{}$.
Formally, it is defined as \textit{a collection of random variables, any finite number of which have a joint Gaussian distribution} \citep{Rasmussen2006}.
The GP is fully specified through its prior mean function $m: \R{d} \rightarrow \R{}, \zeta \mapsto \mathrm{E}[\rho(\zeta)]$ and prior covariance function $\kappa: \R{d} \times \R{d} \rightarrow \R{}, (\zeta, \zeta') \mapsto \mathrm{Cov}[\rho(\zeta), \rho(\zeta')]$, which are design choices \citep{Rasmussen2006}. 

The objective is to learn the underlying function $\varphi$ in terms of inferring a predictive (posterior) distribution for so far unobserved function values (test targets) $\varphi_* = \varphi(Z_*)$, where $Z_* \in \R{n_* \times d}$ is a matrix collecting $n_*$ test inputs $\zeta_*^i \in \R{d}, i = 1, \ldots, n_*$.
To this end, we require a set of (noisy) observations (training targets) $\gamma = \varphi (Z) + \varepsilon$ at training inputs $\zeta^i \in \R{d}, i = 1, \ldots, n$, collected in $Z \in \R{n \times d}$, where $\varepsilon \sim \mathcal{N}(\Zero, \upright{\sigma}{n}^2 \I)$ models white Gaussian noise with variance $\upright{\sigma}{n}^2$.
Therein, $\Zero$ denotes the zero vector and $\I$ the identity matrix. 

By definition, the GP specifies the joint Gaussian distribution of training and test targets, which is the so-called \textit{joint prior distribution}.
Conditioning the joint prior distribution on the training targets yields the predictive (posterior) distribution $\varphi_* \mid Z_*, Z, \gamma \sim \mathcal{N}(m^+(Z_*), \kappa^+(Z_*))$ with
\begin{subequations}
    \begin{align}
        m^+(Z_*) & \! = \! m(Z_*) \! + \! \kappa(Z_*, Z) \kappa_{\gamma}^{-1} (\gamma  \! - \! m(Z)) \label{eq:postmean} \\
        \kappa^+(Z_*) & \! = \! \kappa(Z_*, Z_*) \! + \! \upright{\sigma}{n}^2 \I \! - \! \kappa(Z_*, Z) \kappa_{\gamma}^{-1} \kappa(Z, Z_*) ~, \label{eq:postcov}
    \end{align}
    \label{eq:postdist}%
\end{subequations}
where $\kappa_{\gamma} = \kappa(Z, Z) + \upright{\sigma}{n}^2 \I$.
The posterior mean \eqref{eq:postmean} is an estimate (prediction) for the unobserved function values $\phi_*$; the posterior variances -- the diagonal elements of \eqref{eq:postcov} -- quantify the prediction uncertainty \citep{Rasmussen2006}.

\subsection{Hyperparameter Learning}

The Gaussian process is fully defined via the prior mean function $m( \cdot ; \theta)$ and the prior covariance function $\kappa( \cdot, \cdot ; \theta)$, which usually depend on a set of free parameters $\theta$, the so-called \textit{hyperparameters}.
In consequence, a meaningful predictive distribution \eqref{eq:postdist} requires suitable hyperparameters that adapt the GP model to the underlying problem.

One way to obtain suitable hyperparameters is to \textit{learn} them from the training data.
To this end, we maximize the GP's capability of explaining the training observations, which is expressed by the logarithmic marginal likelihood
\begin{equation*}
    \log ( p(\gamma \mid Z, \theta) ) = -\frac{1}{2} \gamma_0\T \kappa_{\gamma}^{-1} \gamma_0 - \frac{1}{2} \log( \lvert \kappa_{\gamma} \rvert) - \frac{n}{2} \log(2 \pi),
\end{equation*}
where $p(\cdot)$ is a probability density function, $\lvert \cdot \rvert$ denotes the determinant, and $\gamma_0 = \gamma - m(Z)$.
The optimal hyperparameters are determined by maximizing the marginal likelihood w.r.t. the hyperparameters, e.g., by Newton-Raphson method.

\subsection{Prediction at Uncertain Inputs}

So far, we have considered predictions at deterministic test inputs $\zeta_*$.
However, in many applications as well as in the case of multiple-step predictions, the test inputs are subject to stochastic uncertainty.
To this end, we consider in the following predictions of function values $\varphi_* = \varphi(\zeta_*)$ at uncertain test inputs $\zeta_* \sim \mathcal{N}(\mu_*, \Sigma_*)$, where $\mathcal{N}(\mu_*, \Sigma_*)$ denotes a Gaussian distribution with mean vector $\mu_*$ and covariance matrix $\Sigma_*$.

The exact predictive distribution is then obtained via
\begin{equation*}
    p(\varphi(\zeta_*) \mid \mu_*, \Sigma_*, Z, \gamma) = \int p(\varphi(\zeta_*) \mid \zeta_*, Z, \gamma) p(\zeta_*) d \zeta_* ~,
\end{equation*}
where $p(\varphi(\zeta_*) \mid \zeta_*, Z, \gamma)$ denotes the posterior distribution \eqref{eq:postdist}.
Applying the law of iterated expectations (Fubini's theorem) and for a suitable choice of the covariance function $\kappa( \cdot , \cdot )$, the mean $\mu_{\varphi_*}$ and the variance $\sigma_{\varphi_*}^2$ of $p(\varphi(\zeta_*) \mid \mu_*, \Sigma_*, Z, \gamma)$ can be computed analytically as \citep{Deisenroth2010}
\begin{subequations}
    \begin{align}
        \mu_{\varphi_*} & = \mathrm{E}_{\zeta_*} [m^+(\zeta_*) \mid \mu_*, \Sigma_*] ~, \label{eq:postmean_uncertain} \\
        \sigma_{\varphi_*}^2 & = \mathrm{E}_{\zeta_*}[\kappa^+(\zeta_*) \mid \mu_*, \Sigma_*] + \mathrm{Var}_{\zeta_*}[m^+(\zeta) \mid \mu_*, \Sigma_*], \label{eq:postvar_uncertain}
    \end{align}
\end{subequations}
where conditioning on $\mu_*, \Sigma_*$ indicates the uncertainty of $\zeta_*$.
However, since $m^+(\cdot)$ and $\kappa^+(\cdot)$ are (in general) nonlinear functions, computing \eqref{eq:postmean_uncertain}, \eqref{eq:postvar_uncertain} is challenging.
Furthermore, the predictive distribution $p(\varphi(\zeta_*) \! \mid \! \mu_*, \Sigma_*, Z, \gamma)$ is non-Gaussian and analytically intractable when propagating an uncertain input through the nonlinear GP model.

We approximate the exact predictive distribution by linearizing the posterior equations \eqref{eq:postmean}, \eqref{eq:postcov} around $\mu_*$ and apply \eqref{eq:postmean_uncertain} and \eqref{eq:postvar_uncertain}, yielding
\begin{subequations}
    \begin{align}
        \tilde{\mu}_{\varphi_*} & = m^+(\mu_*), \\
        \tilde{\sigma}^2_{\varphi_*} & = \kappa^+(\mu_*) + \nabla_{\zeta_*}\T m^+(\mu_*) \Sigma_* \nabla_{\zeta_*} m^+(\mu_*).
    \end{align}
    \label{eq:lingpmodel}%
\end{subequations}
Propagating the uncertain input $\zeta_*$ through the linearized posterior GP model results then in an analytically tractable Gaussian distribution $\varphi_* \sim \mathcal{N}(\tilde{\mu}_{\varphi_*},  \tilde{\sigma}^2_{\varphi_*})$ \citep{Deisenroth2010, Hewing2017}.

\subsection{Gaussian Process Models of Human Drivers for MPC}
\label{sec:humandrivermodel}

We aim at using GP regression to build models of human driving behavior for intersection crossing.
To this end, we rely on demonstration data from human-driven vehicles turning right, turning left and going straight on over it.
Those three \textit{modes} of the human-driven vehicles are considered separately, such that we compute an independent GP model for each intention \citep{Nguyen2008, Nguyen2009}.
In the following, we show the general derivation of such a model for a single mode.

The objective is to derive a GP-based model for predicting the vehicle position $(p_x(t_i), p_y(t_i))$ with $ i = 1, \ldots, N$ along the MPC horizon given an initial position $(p_x(t_0), p_y(t_0))$.
We rely on a data set $\{ (\zeta_i = (p_x, p_y)_i, \gamma_{x,i} = v_{x,i}, \gamma_{y,i} = v_{y,i}  \mid i = 1, \ldots, n \}$ of two-dimensional Cartesian vehicle positions $(p_x, p_y)$ and corresponding directed velocities $(v_x, v_y)$.
Based thereon, we train two independent GPs to learn the velocity profiles $v_x = \chi_x(p_x, p_y)$ and $v_y = \chi_y(p_x, p_y)$.
The resulting posterior means represent the average velocities while the posterior variances capture variations in the demonstrations.
Given the two GP models of the directed velocities, we iterate the following procedure along the horizon for $i = 0, \ldots, N-1$:
\begin{enumerate}[label=(\roman*)]

    \item Given position $(p_{x,i}, p_{y,i}) \sim \mathcal{N}( \mu_{*,i}, \Sigma_{*,i})$, compute the predictions $v_{x,i} \sim \mathcal{N}(\tilde{\mu}_{x,i}, \tilde{\sigma}_{x,i}^2)$ and $v_{y,i} \sim \mathcal{N}(\tilde{\mu}_{y,i}, \tilde{\sigma}_{y,i}^2))$ according to \eqref{eq:lingpmodel}.\footnote{Note that for $i=0$ we have $\mu_{*,i} = \begin{bmatrix} p_{x,i} & p_{y,i} \end{bmatrix}\T$ and $\Sigma_{*, i} = \Zero$.}
    
    \item Compute the successor position, using the forward Euler integration scheme with sampling time $\upright{T}{s}$, according to $p_{x,i+1} = p_{x,i} + \upright{T}{s} v_{x,i}$, $p_{y,i+1} = p_{y,i} + \upright{T}{s} v_{y,i}$. Then, $(p_{x,i+1}, p_{y,i+1}) \sim \mathcal{N} \big( \mu_{*,i} + \upright{T}{s} \begin{bmatrix} \tilde{\mu}_{x,i} & \tilde{\mu}_{y,i} \end{bmatrix}\T, \Sigma_{*,i} + \upright{T}{s}^2 \mathrm{diag} \big( \begin{bmatrix} \tilde{\sigma}_{x,i}^2 & \tilde{\sigma}_{y,i}^2 \end{bmatrix} \big) \big)$.\footnote{Here, we neglect correlations between $(p_{x,i}, p_{y,i})$ and $v_{x,i}$ and $v_{y,i}$, respectively. As the GPs are independent by design, the covariance matrix of the predicted velocity is diagonal, denoted by $\mathrm{diag(\cdot)}$.}
    
    \item If $i < N-1$, increase $i$ by one and go to step (i).
    
\end{enumerate}

To compute confidence bounds on the predicted positions, we employ Chebyshev's inequality, yielding
\begin{align}
\begin{split}
    P( \lvert p_{x,i} - \mu_{x,*,i} \rvert \leq \nu \sigma_{x,*,i}) \geq 1 - \omega ~, \\
    P( \lvert p_{y,i} - \mu_{y,*,i} \rvert \leq \nu \sigma_{y,*,i}) \geq 1 - \omega ~,
    \end{split}
    \label{eq:chebyshev}%
\end{align}
with maximum probability of constraint violation $\omega \in (0,1)$ and $\nu = \sqrt{\omega^{-1}}$, see \citep{Olkin1958}.
Furthermore, the mean and variance have to be known exactly.
Chebyshev's inequality holds independently of the particular distribution law and tends to conservative estimates of the confidence intervals.
However, we exploit these properties to compensate for distribution mismatches due to the approximated multiple step predictions as well as to compensate for (accumulated) errors in the predicted means and variances.
Thereby, we achieve a certain robustness although it remains open to establish guarantees.

\section{Simulation Example}
\label{sec:SimulationExample}

We start by deriving the GP-based models of human-driving behavior, including a validation of the derived confidence bounds.
Thereafter, we outline the ego vehicle model used in this article and finish by presenting closed-loop simulation results.

\subsection{Modeling of Human Driving Behavior}
\label{subsec:HumanGPModel}

\begin{figure}[b!]
    \centering
    \input{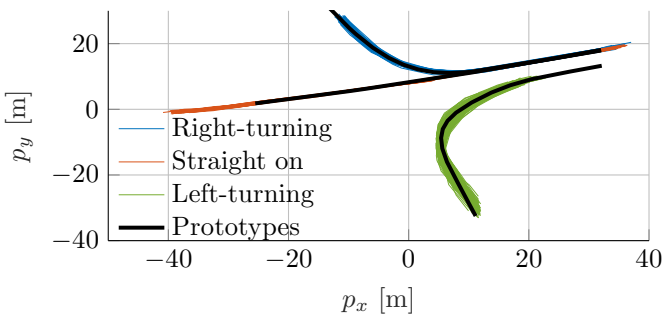}
    \caption{Trajectories of human-driven vehicles crossing the intersection. Left-turning vehicles approach the intersection on a separate lane. Black lines indicate the average paths.}
    \label{fig:intersection_data}
\end{figure}

Given $(p_x, p_y)$-trajectories, with sampling time $\upright{T}{s} = \SI{0.04}{\s}$, of human-driven vehicles crossing an intersection (Fig.~ \ref{fig:intersection_data}), we start by estimating the corresponding vehicle velocities using finite differences.
We define for each intention and each velocity profile a zero-mean GP prior with squared exponential covariance function (\cite{Rasmussen2006}), and train them on uniformly selected subsets of the data of the respective intention.
The selection of a smaller number of active training data points (compared to the number of available data points) reduces the computational complexity of the training as well as of the model evaluations.
We finally exploit the learned GPs and derive the models of human-driven vehicles as described in Sec. \ref{sec:humandrivermodel}.

Having derived these models, we validate them as well as the confidence bounds computed according to Chebyshev's inequality for each intention.
To this end, we choose a horizon length $N$ and assemble from each trajectory, belonging to the respective intention, all horizons that are completely contained in it.
We subsequently evaluate the models along each of the respective horizons and compute for each horizon (exemplarily) the 99\% confidence regions according to \eqref{eq:chebyshev}.
The confidence level is in turn estimated empirically as the relative frequency of trajectories, for which at least 99\% of the corresponding horizons are contained in the appendant confidence regions.
As increasing horizon lengths usually result in a stronger deterioration of the multiple-step prediction quality, we show the empirical confidence levels dependent on the horizon length and, exemplarily, for the right-turning vehicles in Tab. \ref{tab:confidencelevels}.
In accordance with the validation results for all modes, we select -- for the MPC -- a horizon length of $N=40$, for which the designed confidence bounds are exact.

\begin{table}[t]
    \centering
    \caption{Empirical confidence levels of the models for human-driven, right-turning vehicles based on 114 trajectories.}
    \begin{tabular}{c | c | c | c | c | c | c}
        $N$ & 38 & 39 & 40 & 41 & 42 & 43 \\ \hline
        $p_x$-confidence & 1 & 1 & 1 & 0.9912 & 0.9912 & 0.9912 \\ 
        $p_y$-confidence & 1 & 1 & 1 & 0.9912 & 0.9912 & 0.9912 \\ 
        Total confidence & 1 & 1 & 1 & 0.9825 & 0.9825 & 0.9825
    \end{tabular}
    \label{tab:confidencelevels}
\end{table}

\subsection{Ego Vehicle Model}

The ego vehicle's states 
\begin{equation}
    x = [p_x,\, p_y, \,v,\, \psi,\, \delta,\, v_s,\, s]^T
\end{equation}
contain the Cartesian x- and y-positions $p_x$ and $p_y$ respectively, the vehicle's speed $v$, its heading $\psi$, the steering rate $\delta$, the path velocity $v_s$ and the path parameter $s$, see \citep{Kong2015}.
The continuous-time dynamics are given by
\begin{equation}
    \dot{x} = [v \cos{\psi},\, v \sin{\psi},\, u_a, \, \frac{v}{L} \tan \delta,\,u_\delta,\, s,\, u_s]^T
\end{equation}
with vehicle length $L$ and control inputs
\begin{equation}
    u = [u_a,\, u_\delta, \,  u_s]^T,
\end{equation}
where $u_a$ is the acceleration, $u_\delta$ is the steering angle and $u_s$ denotes the virtual input for the path parameter.
The time-discrete state transition map \eqref{eq:SysModel_ego:stateMapping} is therefrom derived via discretization with sampling time $\upright{T}{s} = \SI{0.04}{\s}$.
The output in  \eqref{eq:SysModel_ego:outputMapping} is $y =  [p_x,\,p_y]^T$.


\subsection{Multi-Mode MPC for Safe Intersection Crossing}

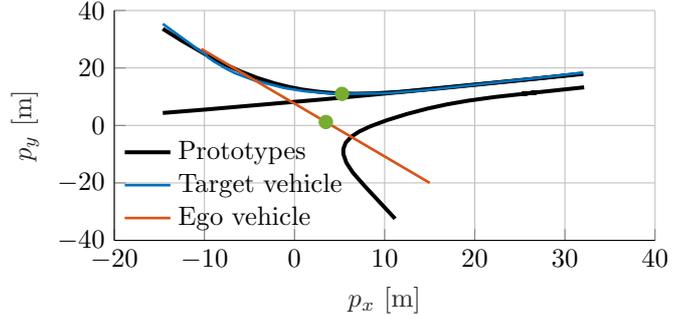
\begin{figure}[t!]
    \centering
%
%

\definecolor{mycolor1}{rgb}{0.00000,0.44700,0.74100}%
\definecolor{mycolor2}{rgb}{0.85000,0.32500,0.09800}%
\definecolor{mycolor3}{rgb}{0.46600,0.67400,0.18800}%

\begin{tikzpicture}

\begin{axis}[%
width=2.8in,
height=1.2in,
at={(0,0)},
scale only axis,
xmin=-20,
xmax=40,
xlabel style={font=\color{white!15!black}},
xlabel={$p_x\text{ [m]}$},
ymin=-40,
ymax=40,
ylabel style={font=\color{white!15!black}},
ylabel={$p_y\text{ [m]}$},
axis background/.style={fill=white},
axis x line*=bottom,
axis y line*=left,
legend style={at={(0,0)}, legend cell align=left, align=left, draw=white!15!black, fill = none,draw = none, anchor = south west},
xmajorgrids,
ymajorgrids
]
\addplot [color=black, line width=1.5pt]
  table[row sep=crcr]{%
32.1315085937367	13.2624914067001\\
25.2020378062598	10.9519621936043\\
26.8901898355265	11.5038101647939\\
22.2026338020206	9.88608981636563\\
19.2685984973297	8.61660753310046\\
16.9390519750409	7.34313897999831\\
14.790035809131	5.88914007051731\\
12.4820089536453	3.99364933804691\\
10.2005007564986	1.77138869095496\\
8.3962937742262	-0.32767065844611\\
7.1822245360344	-2.11611398308007\\
6.29400765877784	-3.93015841236672\\
5.67673067248676	-6.01514273261881\\
5.41315285842012	-8.1535749688128\\
5.41509403499038	-9.95702368307852\\
5.75333888086833	-12.9975298354997\\
6.47281127343352	-16.419263534608\\
7.75614142857081	-21.0053575088537\\
11.0525633075115	-32.1083120511414\\
11.1444603329137	-32.5004603328262\\
};
\addlegendentry{Prototypes}

\addplot [color=black, line width=1.5pt,forget plot]
  table[row sep=crcr]{%
32.0219999999972	18.081715754462\\
26.5520000000251	16.2412261629946\\
11.7320000000764	11.700429649578\\
8.97200000006706	11.1871964294128\\
6.66200000001118	11.0445912880069\\
4.44200000003912	11.2490762518456\\
2.86200000008103	11.6454082355027\\
0.902000000001863	12.4729149815381\\
-0.657999999937601	13.4238817328056\\
-2.24800000002142	14.677452873586\\
-4.10800000000745	16.519821828355\\
-5.92799999995623	18.7147306686648\\
-7.85800000000745	21.4488476845223\\
-9.87800000002608	24.7182911073883\\
-12.5180000000187	29.5103685103168\\
-14.6180000001824	33.6223258058304\\
};

\addplot [color=black, line width=1.5pt, forget plot]
  table[row sep=crcr]{%
32.0219999999972	17.9178629355463\\
16.3220000000438	12.9590241234935\\
3.37199999997392	9.12979117855119\\
-9.87800000002608	5.53591496650623\\
-14.6180000001824	4.34278058736598\\
};

\addplot [color=mycolor1, line width=1.0pt]
  table[row sep=crcr]{%
32.0219999999972	18.3720000004396\\
30.1419999999925	17.6919999998063\\
29.0520000000251	17.2420000005513\\
14.2020000000484	12.1919999998063\\
12.5420000000158	11.8020000001416\\
10.2820000000065	11.2220000000671\\
5.01199999998789	10.9620000002906\\
4.14199999999255	11.1220000004396\\
3.62199999997392	11.2420000005513\\
2.86200000008103	11.4620000002906\\
1.86200000008103	11.7120000002906\\
1.12199999997392	11.9120000004768\\
-1.12800000002608	13.0319999996573\\
-4.33799999998882	15.5719999996945\\
-4.71799999999348	15.9920000005513\\
-6.63799999991897	18.2319999998435\\
-7.57799999997951	19.7719999998808\\
-8.04799999995157	20.5820000004023\\
-8.16799999994691	20.8720000004396\\
-8.98800000001211	22.311999999918\\
-9.20799999998417	22.8720000004396\\
-9.61800000001676	23.7020000005141\\
-9.86800000001676	24.7420000005513\\
-11.0479999999516	27.0820000004023\\
-11.1080000000075	27.3720000004396\\
-14.6180000001824	35.3519999858318\\
};
\addlegendentry{Target vehicle}

\addplot [color=mycolor2, line width=1.0pt]
  table[row sep=crcr]{%
15	-20\\
-10.3524017786527	26.6481342931922\\
};
\addlegendentry{Ego vehicle}

\addplot [color=mycolor3, only marks, mark size=2.5pt, mark=*, mark options={solid, fill=mycolor3, mycolor3}]
  table[row sep=crcr]{%
5.26199999998789	10.9819999998435\\
};

\addplot [color=mycolor3, only marks, mark size=2.5pt, mark=*, mark options={solid, fill=mycolor3, mycolor3}, forget plot]
  table[row sep=crcr]{%
3.47591992120239	1.19812069376644\\
};
\end{axis}
\end{tikzpicture}%
    \caption{Scenario of the simulation example. Green points indicate the vehicle positions at the time at which the actual intention of the target vehicle was identified.}
    \label{fig:Scenario}
\end{figure}

\begin{figure}[b!]
    \centering
%
%

\definecolor{mycolor1}{rgb}{0.00000,0.44700,0.74100}%
\definecolor{mycolor2}{rgb}{0.85000,0.32500,0.09800}%
\definecolor{mycolor3}{rgb}{0.46600,0.67400,0.18800}%

\begin{tikzpicture}

\begin{axis}[%
width=2.8in,
height=1.2in,
at={(0,0)},
scale only axis,
xmin=0,
xmax=7.76,
xlabel style={font=\color{white!15!black}},
xlabel={$t$ [s]},
ymin=0,
ymax=1,
axis x line*=bottom,
axis y line*=left,
ylabel style={font=\color{white!15!black}},
ylabel={$P(i|\{\xi\}(k))$},
axis background/.style={fill=white},
legend style={legend cell align=left, align=left, draw=white!15!black,draw = none,at={(0.5,0.5)},anchor=west},
xmajorgrids,
ymajorgrids
]
\addplot [color=mycolor1, line width=1.5pt]
  table[row sep=crcr]{%
0	0.476537259214817\\
0.04	0.504887600884517\\
0.0800000000000001	0.519485244589284\\
0.16	0.533278944880511\\
0.36	0.548025534903508\\
0.68	0.560363866657161\\
0.8	0.565385765878001\\
1.04	0.585896783553478\\
1.44	0.645301107638981\\
1.48	0.643236146339662\\
1.52	0.63687229286949\\
1.6	0.605286824592753\\
1.84	0.48721426607182\\
1.92	0.458454877404674\\
2	0.437564779272719\\
2.08	0.422254103586493\\
2.16	0.420655447954202\\
2.24	0.424301751347685\\
2.32	0.420995904401183\\
2.4	0.411438345327548\\
2.48	0.393196239482033\\
2.56	0.369674506042836\\
2.64	0.357789296481825\\
2.72	0.373780450000837\\
2.8	0.427206217341473\\
2.92	0.560829591659191\\
3	0.650642060493055\\
3.08	0.72518467337313\\
3.16	0.782437048907672\\
3.2	0.803407824055378\\
3.28	0.833276848759726\\
3.36	0.851827723711528\\
3.44	0.863485669296669\\
3.68	0.87029551122209\\
4.16	0.860280570051876\\
4.4	0.845018751383267\\
4.64	0.835427551037719\\
4.84	0.833364037352708\\
5	0.83167552755984\\
5.36	0.832579792660468\\
5.6	0.84076090739274\\
5.76	0.850994709854852\\
5.96	0.870162307340507\\
6.28	0.906310825256245\\
6.72	0.966864006694211\\
6.84	0.971939999689444\\
7	0.968839867449494\\
7.52	0.929283639926915\\
7.76	0.918290819201022\\
};
\addlegendentry{Right-turning}

\addplot [color=mycolor2, line width=1.5pt]
  table[row sep=crcr]{%
0	0.422840826377894\\
0.0800000000000001	0.420306915328893\\
0.28	0.415391170946884\\
0.96	0.400282418025023\\
1.12	0.384779816393886\\
1.28	0.363381129436548\\
1.4	0.346001834118105\\
1.44	0.34224917744289\\
1.48	0.344929244500546\\
1.52	0.351846762929507\\
1.6	0.384005351600955\\
1.8	0.485259409249424\\
1.92	0.528295879264043\\
2	0.548218920378765\\
2.08	0.562713102434326\\
2.16	0.563748083150414\\
2.24	0.560046752417786\\
2.32	0.564002958334417\\
2.4	0.574526751878603\\
2.48	0.593954850280292\\
2.56	0.618675485166742\\
2.64	0.631437431575592\\
2.72	0.615557212060091\\
2.8	0.561378639481052\\
2.92	0.426106108061504\\
3	0.335759000521434\\
3.08	0.261337101602584\\
3.16	0.204605387133776\\
3.2	0.183860384121147\\
3.28	0.154200147053541\\
3.36	0.135518128599178\\
3.44	0.1234324300125\\
3.68	0.113276998576171\\
4.16	0.113600147252489\\
4.44	0.123075579255355\\
4.68	0.12462256523657\\
5	0.122545989151884\\
5.44	0.114377190892442\\
5.72	0.102448862054652\\
6.28	0.059662755412778\\
6.72	0.0203754506582925\\
6.88	0.0170149226978547\\
7.04	0.0207207494776203\\
7.56	0.0425470742899963\\
7.76	0.0473087114589994\\
};
\addlegendentry{Straight on}

\addplot [color=mycolor3, line width=1.5pt]
  table[row sep=crcr]{%
0	0.100621914407289\\
0.04	0.0722691926272709\\
0.0800000000000001	0.0602078400818229\\
0.16	0.0494209098616345\\
0.36	0.0378153730768958\\
0.76	0.026432767621845\\
1.6	0.0107078238062925\\
2	0.014216300348517\\
2.32	0.0150011372643997\\
2.88	0.0125556188920317\\
3.16	0.0129575639585511\\
3.52	0.0142333215423998\\
4.32	0.0307245980278568\\
4.92	0.0444487871340531\\
5.36	0.0505865549156042\\
5.76	0.0494042690521699\\
6.12	0.0403423625439485\\
6.92	0.0111098254686777\\
7.16	0.0179599200150786\\
7.64	0.0320364773637287\\
7.76	0.0344004693399791\\
};
\addlegendentry{Left-turning}

\addplot [color=black, line width=1.5pt]
  table[row sep=crcr]{%
0	0.15\\
7.76	0.15\\
};
\addlegendentry{$\epsilon_{\mathrm{P}}$}

\end{axis}


\end{tikzpicture}%
%
    \caption{Online classification results. Modes with probability below $\upright{\epsilon}{P}$ are considered inactive.}
    \label{fig:classification}
\end{figure}
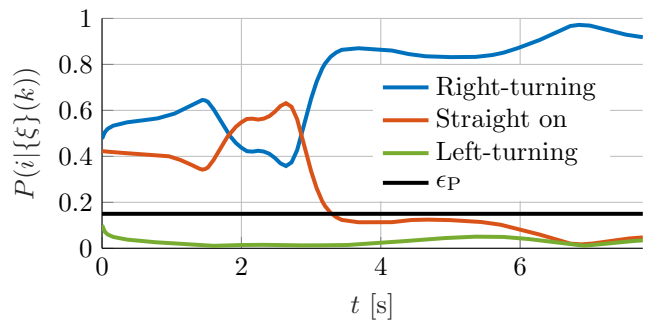

To illustrate our approach, we select a specific trajectory from the data set to represent the target vehicle. 
We depict the situation in Fig.~\ref{fig:Scenario}.
The target vehicle starts east and takes a right turn, while the ego vehicle starts south and drives straight on over the intersection.

We employ the control scheme presented in Sec. \ref{sec:MultiModeMPC} to safely cross the intersection with the ego vehicle.
Initially, the intention of the target vehicle is unknown to the ego vehicle.
Consequently, the ego vehicle needs to  consider all possible intentions of the target vehicle.
However, as the target vehicle progresses, we continuously employ the online classification algorithm, c.f., Sec. \ref{sec:MultiModeMPC}, with $\epsilon_{\mathrm{P}} = 0.15$, to estimate the actual intention of the target vehicle.
In Fig.~\ref{fig:classification}, we show the result of the classification algorithm at each sampling time point.
Since vehicles that turn left already start on another lane, this intention immediately falls below the threshold $\epsilon_{\mathrm{P}}$ and is thus considered inactive.
However, going straight on or turning right are still possible.
After the target vehicle starts turning right, i.e., after about $\SI{3.4}{\second}$, we can identify the intention of the target vehicle.
The corresponding positions of the vehicles at that time are marked by green points in Fig.~\ref{fig:Scenario}. 

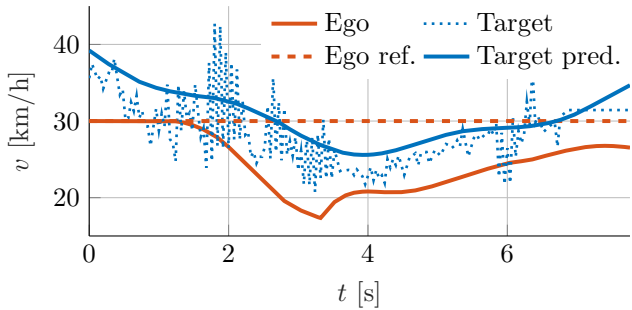
\begin{figure}
    \centering
%
%

\definecolor{mycolor1}{rgb}{0.00000,0.44700,0.74100}%
\definecolor{mycolor2}{rgb}{0.85000,0.32500,0.09800}%

\begin{tikzpicture}

\begin{axis}[%
width=2.8in,
height=1.2in,
at={(0,0)},
scale only axis,
xmin=0,
xmax=7.76,
xlabel style={font=\color{white!15!black}},
xlabel={$t$ [s]},
ymin=15,
ymax=45,
ylabel style={font=\color{white!15!black}},
ylabel={$v$ [km/h]},
axis background/.style={fill=white},
axis x line*=bottom,
axis y line*=left,
legend style={at={(1.01,1.02)}, legend cell align=left, align=left, draw=white!15!black, draw = none, anchor = north east, legend columns=2},
xmajorgrids,
ymajorgrids
]
\addplot [color=mycolor2, line width=1.5pt]
  table[row sep=crcr]{%
0	30\\
1.24	29.9677114745022\\
1.36	29.7932587020773\\
1.48	29.4838867568397\\
1.64	28.8901158093529\\
1.84	27.8515342767123\\
2	26.6123979385631\\
2.8	19.566134725041\\
3.04	18.3190886715764\\
3.28	17.4593041852108\\
3.32	17.3499461503297\\
3.52	19.428213135002\\
3.68	20.301294753817\\
3.8	20.6429483463218\\
3.92	20.7814153569722\\
4	20.8100873325884\\
4.12	20.7711693778175\\
4.24	20.7138117920202\\
4.48	20.7219741301161\\
4.68	20.9460772380479\\
4.92	21.4783682962718\\
5.32	22.6205323682592\\
5.72	23.7758225290079\\
6.04	24.558044376936\\
6.2	24.7559685655845\\
6.36	24.994037464748\\
6.56	25.4913266644502\\
6.84	26.1318694268332\\
7.12	26.5992331123404\\
7.28	26.7361653153854\\
7.4	26.7654793310484\\
7.52	26.7367222644061\\
7.68	26.6196380227351\\
7.76	26.5324862566058\\
};
\addlegendentry{Ego}

\addplot [color=mycolor1, dotted, line width=1.0pt]
  table[row sep=crcr]{%
0	35.6040728087502\\
0.0399999999999991	37.466384930524\\
0.119999999999997	35.4558880899899\\
0.159999999999997	35.8647459239539\\
0.200000000000003	34.2118400493773\\
0.280000000000001	36.2801598382212\\
0.32	36.1459541604997\\
0.359999999999999	38.1413161719278\\
0.439999999999998	33.4456275270307\\
0.479999999999997	34.152598731061\\
0.520000000000003	30.1734983005736\\
0.560000000000002	30.3608300293926\\
0.600000000000001	32.4499614790488\\
0.640000000000001	31.3065488300759\\
0.68	33.3000000005531\\
0.719999999999999	30.4540637730244\\
0.759999999999998	31.3065488243812\\
0.799999999999997	29.3156954644388\\
0.840000000000003	29.3156954486775\\
0.880000000000003	26.7589237540051\\
0.960000000000001	31.6026897638942\\
1	28.7577815693017\\
1.04	30.9945156164148\\
1.08	29.3156954744102\\
1.12	34.7404663016668\\
1.16	27.6081509892819\\
1.2	27.0748961160085\\
1.24	24.7622696818204\\
1.28	34.7404663016668\\
1.32	28.7577815745327\\
1.36	33.8787544113694\\
1.4	30.4540637457765\\
1.44	31.0858488714786\\
1.48	28.7577815693017\\
1.52	34.763774235623\\
1.6	27.9145123381879\\
1.64	29.0520223005032\\
1.68	23.9135944833821\\
1.72	37.6066483306962\\
1.76	24.7622696818204\\
1.8	42.6907484179678\\
1.84	26.7589237356807\\
1.88	42.4147380256261\\
1.92	26.4697940934724\\
1.96	38.991922239649\\
2	26.4697941034472\\
2.04	37.5958774369127\\
2.08	27.6081509793768\\
2.12	36.7129405004365\\
2.16	27.6081509619576\\
2.2	32.7110073174297\\
2.24	28.1888275635682\\
2.28	28.6024474579164\\
2.32	25.3442301201609\\
2.36	25.9755654497691\\
2.4	24.4826469079485\\
2.44	25.3442300917388\\
2.48	25.3442301155508\\
2.52	26.652767223141\\
2.56	31.9595056192221\\
2.6	25.1033862335725\\
2.64	36.1011080258605\\
2.68	25.9755654192754\\
2.72	32.4124975976149\\
2.76	24.8927700357384\\
2.8	31.615502536116\\
2.84	25.5986327646457\\
2.88	27.3724313871225\\
2.92	25.4558441251235\\
2.96	27.8999999990178\\
3	23.555254193953\\
3.04	26.1155126334719\\
3.08	21.7680959116138\\
3.12	27.9580042173753\\
3.16	21.6748702371765\\
3.2	27.9000000044201\\
3.24	20.7781134867596\\
3.28	28.2605732442724\\
3.32	22.5718851675767\\
3.36	30.1734982905731\\
3.4	22.5179928044587\\
3.44	27.7252592509198\\
3.48	22.5179928044587\\
3.52	26.2392835391685\\
3.6	23.4691286519493\\
3.64	23.6239285529572\\
3.68	23.555254193953\\
3.72	23.6239285231141\\
3.76	23.555254193953\\
3.8	22.9455878221735\\
3.84	23.6753035891601\\
3.92	22.7861800221422\\
3.96	21.7494827397887\\
4.04	23.9643485365365\\
4.08	22.2647703680914\\
4.12	23.6239285470634\\
4.16	22.5000000062144\\
4.2	22.9102596868291\\
4.24	21.9164321915358\\
4.28	25.0225897956787\\
4.32	23.9135944593409\\
4.36	25.022589838639\\
4.4	24.2332416191245\\
4.44	25.4876440870123\\
4.52	22.9455878102557\\
4.6	25.3602050412957\\
4.64	25.4876440705603\\
4.68	25.7877878330578\\
4.72	24.2165232304373\\
4.76	25.7877878330578\\
4.8	25.9599691313977\\
4.84	25.9599691895158\\
4.88	24.4826469657219\\
4.96	26.2392835320821\\
5.04	25.4876440588947\\
5.12	27.0000000119835\\
5.16	26.8495809705445\\
5.2	26.469794123753\\
5.24	26.7589237540051\\
5.28	27.4758075242581\\
5.32	26.7286363283885\\
5.36	28.9822359415216\\
5.44	25.3602050461295\\
5.52	27.7982013459057\\
5.56	28.2462386899766\\
5.6	27.5494101690745\\
5.68	28.2605732389337\\
5.72	28.188827559554\\
5.76	28.9822359864212\\
5.8	28.2462386899766\\
5.84	29.7136332041784\\
5.88	22.7149730328351\\
5.92	28.4604989210148\\
5.96	21.7494828100249\\
6	26.4697940553546\\
6.04	21.8979451289897\\
6.08	30.4540637841712\\
6.12	25.3442301356421\\
6.16	28.7155010077152\\
6.2	27.0748961233224\\
6.24	28.246238726031\\
6.28	32.524759765355\\
6.32	29.4810108642244\\
6.36	36.9109739905596\\
6.4	29.8088912162732\\
6.44	30.600000059879\\
6.52	28.5882842608519\\
6.64	29.6043915356786\\
6.68	29.79530170369\\
6.72	26.6527672146499\\
6.76	30.6132323992885\\
6.88	31.435648501431\\
7.76	31.435648501431\\
};
\addlegendentry{Target}

\addplot [color=mycolor2, dashed, line width=1.5pt]
  table[row sep=crcr]{%
0	30\\
7.76	30\\
};
\addlegendentry{Ego ref.}

\addplot [color=mycolor1, line width=1.5pt]
  table[row sep=crcr]{%
0	39.2099696834471\\
0.280000000000001	37.4307102018056\\
0.719999999999999	35.1269156958976\\
0.960000000000001	34.339533750954\\
1.16	33.8813157823165\\
1.52	33.2977891313239\\
1.8	32.9982923088119\\
2.12	32.248727457322\\
2.28	31.6549727771743\\
2.72	29.8473466049893\\
2.88	28.9241997656765\\
3.04	28.0921201095565\\
3.32	26.8666615354504\\
3.64	25.8661441881713\\
3.8	25.634312789206\\
3.92	25.5818510196448\\
4	25.5908347094813\\
4.08	25.6442386994714\\
4.2	25.7842055472447\\
4.4	26.1901408772177\\
4.64	26.8479913653001\\
5.12	28.1971321513528\\
5.36	28.6732546695609\\
5.64	28.9519025463731\\
5.96	29.1092046158953\\
6.16	29.1938678097941\\
6.32	29.3272593429295\\
6.4	29.4430846719425\\
6.6	29.7564971830105\\
6.76	30.1444934284934\\
6.96	30.830301672542\\
7.2	31.9355181321329\\
7.76	34.6756039694246\\
};
\addlegendentry{Target pred.}

\end{axis}

\end{tikzpicture}%
    \caption{Actual velocity profiles of the ego and target vehicle and GP-based prediction of the target vehicle speed. Additionally, the velocity reference for the ego vehicle is shown.}
    \label{fig:velocities}
\end{figure}

In order to keep the desired minimal safety distance\footnote{The desired minimal safety distance is dependent on the ego vehicle speed and chosen as the distance that is covered by the ego vehicle in $\SI{1}{\s}$.}, the ego vehicle decelerates such that the target vehicle can pass the intersection first (Fig. \ref{fig:velocities}). 
Afterwards, the ego vehicle accelerates again while keeping the desired safety distance robustly. 
Note that the velocity of the target vehicle is noisy since it is based on an observed trajectory in the data set and thus includes measurement noise.

Finally, we depict in Fig.~\ref{fig:distances} the distance between target and ego vehicle during intersection crossing in more detail.
Therein, the actual distance (red line) between ego and target vehicle and the velocity-dependent, minimal desired safety distance $\upright{d}{safe}$ (black line) are shown.
Moreover, for every $20^{\text{th}}$ iteration, we show the predicted minimal safety distances (black dotted lines) according to the planned ego vehicle speeds, the predicted distances (blue dashed lines) and the tightening (blue dotted lines) in \eqref{eq:OCP_MMConstrLearn:distance} by $d_{\sigma,i}$ over the controller's prediction horizon.
We observe that the predicted distances always fulfill the corresponding tightened distance constraints, indicating that the ego vehicle always keeps enough space to account for possible prediction errors of the target vehicle's behavior.

\begin{figure}
    \centering
    \input{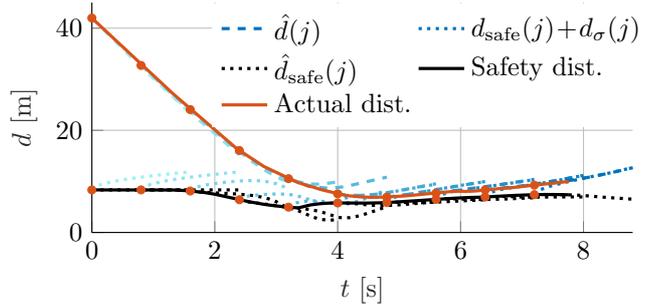}%
    \caption{Actual and velocity-dependent, minimal safety distance between ego and target vehicle. At selected times, the minimal safety distances corresponding to the planned ego vehicle speeds ($\hat{d}_{\mathrm{safe}}$), the predicted distances ($\hat{d}$) and the tightened distance constraints ($d_{\mathrm{safe}}+d_{\sigma}$) are shown over the corresponding prediction horizon of length $N$.}
    \label{fig:distances}
\end{figure}



\section{Conclusions}
\label{sec:Conclusions}

We have presented a multi-mode MPC scheme for safe intersection crossing of autonomous vehicles in the presence of human-driven vehicles.
The proposed MPC scheme is supported by learned, GP-based models of human drivers and an online classification algorithm to determine the human's a-priori unknown intention.
We have described the derivation of such GP-based models as well as their embedding in the multi-mode prediction approach used in the MPC scheme.
Furthermore, we have derived probabilistic safety guarantees for the proposed scheme, combining the multi-mode predictions and GP-based confidence sets of human driving behavior.  
The proposed approach has been illustrated in a simulation example using real-world traffic data.

Future steps will focus on computationally efficient implementations of the GP evaluations to achieve real-time capability of the proposed approach.
Furthermore, we aim to extend the presented approach to consider multiple human-driven vehicles crossing the intersection and investigate the transferability to other arbitrary intersections.
Finally, we will focus on employing more sophisticated modeling approaches such as structured output GPs or intention-driven dynamics models based on GPs.



\bibliography{root}           


\end{document}